\documentclass[oneside,10pt,reqno]{amsart}
\usepackage{latexsym,amsmath,amsthm,amssymb,mathrsfs,bbm}
\usepackage{url}
\usepackage[left=3cm,right=3cm,top=2cm]{geometry}
\usepackage{ifthen}
\usepackage{tikz}
\usepackage{enumitem}
\usepackage[all,2cell]{xy}
\usepackage{csquotes}
\usepackage{soul}
\usepackage{array}
\usepackage{booktabs}

\UseAllTwocells
\SelectTips{eu}{10}

\theoremstyle{definition}
\numberwithin{equation}{section}
\newtheorem{proposition}[equation]{Proposition}
\newtheorem{claim}[equation]{Claim}
\newtheorem{lemma}[equation]{Lemma}
\newtheorem{theorem}[equation]{Theorem}
\newtheorem{corollary}[equation]{Corollary}
\newtheorem{definition}[equation]{Definition}
\newtheorem{remark}[equation]{Remark}
\newtheorem{example}[equation]{Example}
\newtheorem{observation}[equation]{Observation}

\newcommand{\bdc}{\operatorname{bdc}}
\newcommand{\dc}{\operatorname{dc}}
\newcommand{\aS}{\mathfrak{S}}
\newcommand{\per}{\textup{per}}
\newcommand{\weight}{\textup{weight}}
\renewcommand{\det}{\textup{det}}

\newcommand{\gO}{\ensuremath{\mathcal{O}}}

\newcommand{\ComCla}[1]{\ensuremath{\textup{\textbf{{#1}}}}}

\newcommand{\VP}{\ComCla{VP}}
\newcommand{\VNP}{\ComCla{VNP}}
\newcommand{\VPws}{\ComCla{VP}\ensuremath{_{\ComCla{ws}}}}

\newcommand{\VPs}{\ComCla{VP}\ensuremath{_{\ComCla{s}}}}
\newcommand{\deto}{\ensuremath{\ComCla{DETP}^0}}

\title[Binary Determinantal Complexity]{Binary Determinantal Complexity}
\author{Jesko H\"uttenhain$^\ast$}
\thanks{$^\ast$Technische Universit\"at Berlin, jesko$@$math.tu-berlin.de}
\author{Christian Ikenmeyer$^\dagger$}
\thanks{$^\dagger$Texas A\&M University, ciken$@$math.tamu.edu}

\newcommand{\HB}{\mathbb{B}}

\newcommand{\HZ}{\mathbb{Z}}

\newcommand{\HN}{\mathbb{N}}
\newcommand{\kS}{\mathfrak{S}}

\newcommand{\df}{:=}
\let\dd\emph

\makeatletter
\def\ifempty#1{\def\@temp{#1}\ifx\@temp\@empty}
\makeatother
\newcommand{\ifnonempty}[2]{\ifempty{#1}\else#2\fi}
\newcommand{\enclspacing}{}

\newcommand{\cenclose}[7][auto]{%
\ifempty{#1} %
\ifnonempty{#2}{#2\enclspacing}#3 %
\ifnonempty{#4}{#7#4#7}#5%
\ifnonempty{#6}{\enclspacing#6}%
\else\ifthenelse{\equal{#1}{auto}}%
{
\ifthenelse{\equal{#2}{}}{\left.}{\left#2}\enclspacing#3%
\ifthenelse{\equal{#4}{}}{}{#7\middle #4#7}#5%
\enclspacing\ifthenelse{\equal{#6}{}}{\right.}{\right#6}
}{
\ifthenelse{\equal{#2}{}}{}{\csname#1l\endcsname#2}\enclspacing#3%
\ifthenelse{\equal{#4}{}}{}{#7\csname#1\endcsname#4#7}#5%
\enclspacing\ifthenelse{\equal{#6}{}}{}{\csname#1r\endcsname#6}
}
\fi}
\newcommand{\enclose}[4][auto]{\cenclose[#1]{#2}{#3}{}{}{#4}{}}

\newcommand{\set}[2][auto]{\enclose[#1]\{{#2}\}}

\newcommand{\abs}[2][auto]{\enclose[#1]|{#2}|}

\newcommand{\cset}[3][auto]{\cenclose[#1]\{{#2}\lvert{#3}\}\:}

\ifdefined\divides\else \newcommand{\divides}{\mathbin\vert} \fi
\newcommand{\dotter}[3][]{#2#3\ifthenelse{\equal{#1}{}}{}{\widehat{#1}#3}#2}
\newcommand{\pts}[2][]{\dotter[#1]{#2}{\cdots}}
\newcommand{\dts}[2][]{\dotter[#1]{#2}{\ldots}}

\begin{document}
\raggedbottom
\begin{abstract}
We prove that for writing the 3 by 3 permanent polynomial as a determinant of a matrix consisting only of zeros, ones, and variables as entries,
a 7 by 7 matrix is required.
Our proof is computer based and uses the enumeration of bipartite graphs.

Furthermore, we analyze sequences of polynomials that are determinants of polynomially sized matrices consisting only of zeros, ones, and variables.
We show that these are exactly the sequences in the complexity class of constant free polynomially sized (weakly) skew circuits.

\end{abstract}

\maketitle

{\footnotesize
\\[1ex]
\\[1ex]
\noindent\textbf{2012 ACM Computing Classification System:} Theory of computation -- Models of computation; Theory of computation -- Computational complexity and cryptography -- Circuit complexity
}

\section{Introduction}
Let $\aS_m$ denote the symmetric group on $m$ letters and let
$
\per_m := \sum_{\pi \in \aS_m} \prod_{i=1}^m x_{i,\pi(i)}
$
denote the $m \times m$ permanent polynomial in $m^2$ variables.
The flagship problem in algebraic complexity theory is finding superpolynomial lower bounds for
the determinantal complexity of the permanent polynomial,
a question whose roots date back to Valiant's seminal paper \cite{Val:79b}, with an additional emphasis on the special role of the permanent in \cite{val:79}.

We call a matrix whose entries are only variables or integers an \dd{integer variable matrix}.
One main implication of \cite{Val:79b} is the following theorem.
\begin{theorem}\label{thm:universalityofdet}
For every polynomial $f$ with rational coefficients one can always find a square matrix~$A$ whose entries are variables or rational numbers such that $\det(A)=f$.
Moreover, if $f$ has only integer coefficients, then $A$ can be chosen as an integer variable matrix.\qed
\end{theorem}
For example,
\begin{equation}\label{eqn:firstexample}
\det\begin{pmatrix}
0      &x_{11}&x_{21}\\
x_{12}&   0   &   1   \\
x_{22}&   1   &   0
\end{pmatrix} = x_{11} x_{22} + x_{12} x_{21} = \per_2.
\end{equation}
For an $n \times n$ square matrix we refer to $n$ as its \dd{size}.
What is the minimal size of a matrix whose determinant is $\per_m$ and whose entries are only variables and rational numbers?
For a given $m$ we take $\dc(\per_m)$ to be this minimal size. It is famously conjectured by Valiant that the sequence $m \mapsto \dc(\per_m)$ of natural numbers grows superpolynomially fast.
In modern terms we can concisely phrase this conjecture as $\VPws \neq \VNP$, see for example \cite{mapo:08}.
A graph construction by Grenet \cite{Gre:11}, see Section~\ref{grenetconstruction}, has the following consequence.
\begin{theorem}\label{thm:grenet}
For every natural number $m$ there exists an integer variable matrix~$A$ of size $2^m-1$ such that $\per_m = \det(A)$.
Moreover, $A$ can be chosen such that the entries in $A$ are only variables, zeros, and ones, but no other constants. \qed
\end{theorem}
Theorem~\ref{thm:grenet} gives rise to the following definition.
We call a matrix whose entries are only zeros, ones, or variables, a \emph{binary variable matrix}.
We will prove in Corollary~\ref{cor:valuniversality} that every polynomial $f$ with integer coefficients can be written as the determinant of a binary variable matrix
and that the size is almost the size of the matrix from Theorem~\ref{thm:universalityofdet}, see Proposition~\ref{pro:integercoeffs} for a precise statement.
We then denote by $\bdc(f)$ the smallest $n$ such that $f$ can be written as a determinant of an $n \times n$ binary variable matrix.
It turns out that the complexity class of sequences $(f_m)$ with polynomially bounded binary determinantal complexity $\bdc(f_m)$ is exactly $\VPws^0$,
the constant free version of $\VPws$, see Section~\ref{sec:vpws} for definitions and proofs.

Theorem~\ref{thm:grenet} shows that $\bdc(\per_m) \leq 2^m -1$. It is easy to see that this upper bound is sharp for $m=1$ and for $m=2$.

The best known general lower bound is $\bdc(\per_m) \geq \frac {m^2} 2$ due to \cite{MR:04} in a stronger model of computation, see also \cite{lamare:10} for the same bound in an even stronger model of computation.
This implies that $\bdc(\per_3)$ is either~$5$,~$6$, or~$7$.

The main result of this paper is the following.
\begin{theorem}
\label{PermanentLowerBound}
$\bdc(\per_3)=7$.
\end{theorem}
We use a computer aided proof and enumeration of bipartite graphs in our study.
The binary determinantal complexity of $\per_m$ is now known to be exactly $2^m-1$ for $m\in \{1,2,3\}$.
Unfortunately, determining $\bdc(\per_4)$ is currently out of reach with our methods.

\subsection*{Acknowledgments}
We thank Peter B\"urgisser, Gordon Royle, and JM Landsberg very much for interesting and helpful discussions.
We are very grateful to the Simon's Institute for the Theory of Computing in Berkeley for hosting us during this project.

\subsection*{Ancillary files} Since the proof of Theorem \ref{PermanentLowerBound} is computer-based, you can find attached to this note the source code of a C programm called \texttt{ptest}, which performs the algorithm outlined in Section~\ref{stepwise}.

We also include the file \texttt{output-ptest-on-7x7.txt} containing the output of \texttt{ptest} on the \texttt{nauty}-based enumeration of all $7\times 7$ binary variable matrices. This list will be relevant for the observations made in Section \ref{sec:uniqueness}.

\section{Binary Algebraic Branching Programs and the Cost of Computing Integers}

The main purpose of this section is to prove that even though we only allow the constants $0$ and~$1$, \emph{all} polynomials with integer coefficients can be obtained as the determinant of a binary variable matrix, see Corollary~\ref{cor:valuniversality}.
Moreover, the size of the matrices is not much larger than the size of matrices from Theorem~\ref{thm:universalityofdet}, see Proposition~\ref{pro:integercoeffs}.
We use standard techniques from algebraic complexity theory, heavily based on \cite{Val:79b}, but a certain attention to the signs has to be taken.

In what follows, a \dd{digraph} is always a finite directed graph which may possibly have loops, but which has no parallel edges.
We label the edges of a digraph by polynomials. We will almost exclusively be concerned with digraphs whose labels are only variables or the constant~$1$.
Note that we consider only labeled digraphs.

A cycle cover of a digraph $G$ is a set of cycles in $G$ such that each vertex of $G$ is contained in exactly one of these cycles. If a cycle in $G$ has $i$ edges with labels $e_1\dts,e_i$, then its \dd{weight} is defined as $(-1)^{i-1}\cdot e_1\cdots e_i$. The weight of a cycle cover is the product of the weights of its cycles. The \dd{value} of $G$ is the polynomial that arises as the sum over the weights of all cycle covers in~$G$. We then define the \emph{directed adjacency matrix} $A$ of a digraph G as the matrix whose entry $A_{ij}$ is the label of the edge $(i,j)$ or $0$ if that edge does not exist.

In what follows, we will often construct matrices as the directed adjacency matrices of digraphs. The reason is the following well-known observation, see for example \cite{Val:79b}.
\begin{observation}
The value of a digraph $G$ equals the determinant of its directed adjacency matrix.
\end{observation}

As an intermediate step, we will often construct a \emph{binary algebraic branching program}: This is an acyclic digraph $\Gamma=(\Gamma,s,t)$ where every edge is labeled by either $1$ or a variable. The digraph $\Gamma$ has two distinguished vertices, the \emph{source} $s$ and the \emph{target}~$t$, where $s$ has no incoming and $t$ has no outgoing edges. If an $s$-$t$-path in $\Gamma$ has $i$ edges with labels $e_1\dts,e_i$, then its \dd{path weight} is defined as the value $(-1)^{i-1}\cdot e_1\cdots e_i$. The \dd{path value} of $\Gamma$ is the polynomial that arises as the sum over the path weights of all $s$-$t$-paths in~$\Gamma$. We remark that this notion of weight differs from the literature by a sign.

\begin{proposition} \label{prop:constants}
For a nonzero constant $c\in\HZ$, there is a binary algebraic branching program $\Gamma$ with at most $\gO(\log |c|)$ vertices whose path value is~$c$.
\end{proposition}
\begin{proof}
We can assume without loss of generality that $c>0$: Given a binary algebraic branching program $\Gamma$ with path value $c>0$ and at most $\gO(\log c)$ vertices, we can add a single vertex $t'$ and an edge from $t$ to $t'$ with label $1$ to obtain a new program $(\Gamma',s,t')$ with path value~$-c$.

For a natural number~$c$, an \emph{addition chain} of length $\ell$ is a sequence of distinct natural numbers $1=c_0,c_1\dts,c_\ell=c$ together with a sequence of tuples $(j_1,k_1)\dts,(j_\ell,k_\ell)$ such that $c_i=c_{j_i}+c_{k_i}$ and $j_i, k_i < i$ for all $1\le i\le\ell$. However, we will think of this data as a digraph $\tilde \Gamma$ on the vertices $\set{c_0\dts,c_\ell}$ with edges $(c_{j_i},c_i)$ and $(c_{k_i},c_i)$ for all $1\le i\le \ell$. The labels of all edges are equal to~$1$. Note that we allow double edges in these digraphs temporarily.
We set $s\df c_0$ and $t\df c_\ell$. Thus, we view an addition chain as an acyclic digraph where every vertex except for $c_0$ has indegree two. This already strongly resembles a binary algebraic branching program, but $\tilde \Gamma$ might have parallel edges. Observe that there are exactly $c_i$ many paths from $c_0$ to $c_i$ in the digraph $\tilde \Gamma$. In particular, there are exactly $c$ paths from $s$ to $t$ in $\tilde \Gamma$.

Using the algorithm of repeated squaring \cite[Sec.~4.6.3, eq.~(10)]{Knu:97} one can construct an addition chain $\tilde \Gamma$ as above with at most $\gO(\log c)$ vertices and such that there are exactly $c$ paths from $s$ to $t$ in $\tilde \Gamma$. For every edge $(v,w)$ in~$\tilde \Gamma$ we add a new vertex $u$ and replace the edge $(v,w)$ by two new edges $(v,u)$ and $(u,w)$. We call the resulting digraph $\Gamma=(\Gamma,s,t)$. Observe that the binary algebraic branching program $\Gamma$ has no parallel edges any more and all $s$-$t$-paths in $\Gamma$ have even length. Also, the digraph $\Gamma$ still has $\gO(\log c)$ many vertices. Labelling all edges in $\Gamma$ with~$1$, the path value of $\Gamma$ is equal to~$c$.
\end{proof}

\begin{proposition}\label{pro:integercoeffs}
Let $C$ be an $n\times n$ matrix whose entries are variables and \emph{arbitrary} integer entries. Let $c_{\max}$ be the integer entry of~$C$ with the largest absolute value. Then there is a binary variable matrix~$A$ of size $\gO(n^2\cdot\log|c_{\max}|)$ with $\det(A)=\det(C)$.
\end{proposition}
\begin{proof} \tikzset{>=latex}
\begin{figure}[b]
\begin{tikzpicture}[scale=5.2]
\node[draw, rectangle] (3cycles) at (0,1) {
\begin{tikzpicture}[scale=1]
\node[draw, circle,inner sep=0pt,minimum size=0.2cm] (1) at (1,0) {};
\node[draw, circle,inner sep=0pt,minimum size=0.2cm] (2) at (2,0) {};
\node[draw, circle,inner sep=0pt,minimum size=0.2cm] (3) at (4,0) {};
\path (1) edge[ultra thick,out=315,in=225,loop] node[below] {$x$} (1);
\path (2) edge[ultra thick,out=315,in=225,loop] node[below] {$3$} (2);
\path (3) edge[ultra thick,in=315,out=225,loop] node[below] {$y$} (3);
\path (2) edge[->,bend left] node[above] {$-2$} (3);
\path (3) edge[->,bend left] node[below] {$x$} (2);
\end{tikzpicture}
};
\node[draw, rectangle] (2cycles) at (0,2.05) {
\begin{tikzpicture}[scale=1]
\node[draw, circle,inner sep=0pt,minimum size=0.2cm] (1) at (1,0) {};
\node[draw, circle,inner sep=0pt,minimum size=0.2cm] (2) at (2,0) {};
\node[draw, circle,inner sep=0pt,minimum size=0.2cm] (3) at (4,0) {};
\path (1) edge[ultra thick,out=315,in=225,loop] node[below] {$x$} (1);
\path (2) edge[out=315,in=225,loop] node[below] {$3$} (2);
\path (3) edge[in=315,out=225,loop] node[below] {$y$} (3);
\path (2) edge[ultra thick,->,bend left] node[above] {$-2$} (3);
\path (3) edge[ultra thick,->,bend left] node[below] {$x$} (2);
\end{tikzpicture}
};
\node[draw, rectangle] (3cyclesa) at (1,1) {
\begin{tikzpicture}[scale=1]
\node[draw, circle,inner sep=0pt,minimum size=0.2cm] (1) at (1,0) {};
\node[draw, circle,inner sep=0pt,minimum size=0.2cm] (2) at (2,0) {};
\node[draw, circle,inner sep=0pt,minimum size=0.2cm] (3) at (5,0) {};
\node[draw, circle,inner sep=0pt,minimum size=0.2cm] (a) at (2,1) {};
\node[draw, circle,inner sep=0pt,minimum size=0.2cm] (b1) at (3,1.5) {};
\node[draw, circle,inner sep=0pt,minimum size=0.2cm] (b2) at (3,0.5) {};
\node[draw, circle,inner sep=0pt,minimum size=0.2cm] (c) at (4,1) {};
\node[draw, circle,inner sep=0pt,minimum size=0.2cm] (d) at (5,1) {};
\path (1) edge[ultra thick,out=315,in=225,loop] node[below] {$x$} (1);
\path (2) edge[ultra thick,out=315,in=225,loop] node[below] {$3$} (2);
\path (3) edge[ultra thick,in=315,out=225,loop] node[below] {$y$} (3);
\path (2) edge[->] (a);
\path (a) edge[->] (b1);
\path (a) edge[->] (b2);
\path (b1) edge[->] (c);
\path (b2) edge[->] (c);
\path (c) edge[->] (d);
\path (d) edge[->] (3);
\path (3) edge[bend left,->] node[below] {$x$} (2);
\path (a) edge[ultra thick,out=135,in=45,loop] (a);
\path (b1) edge[ultra thick,out=135,in=45,loop] (b1);
\path (b2) edge[ultra thick,out=315,in=225,loop] (b2);
\path (c) edge[ultra thick,out=135,in=45,loop] (c);
\path (d) edge[ultra thick,out=135,in=45,loop] (d);
\end{tikzpicture}
};
\node[draw, rectangle] (2cyclesb) at (1,1.7) {
\begin{tikzpicture}[scale=1]
\node[draw, circle,inner sep=0pt,minimum size=0.2cm] (1) at (1,0) {};
\node[draw, circle,inner sep=0pt,minimum size=0.2cm] (2) at (2,0) {};
\node[draw, circle,inner sep=0pt,minimum size=0.2cm] (3) at (5,0) {};
\node[draw, circle,inner sep=0pt,minimum size=0.2cm] (a) at (2,1) {};
\node[draw, circle,inner sep=0pt,minimum size=0.2cm] (b1) at (3,1.5) {};
\node[draw, circle,inner sep=0pt,minimum size=0.2cm] (b2) at (3,0.5) {};
\node[draw, circle,inner sep=0pt,minimum size=0.2cm] (c) at (4,1) {};
\node[draw, circle,inner sep=0pt,minimum size=0.2cm] (d) at (5,1) {};
\path (1) edge[ultra thick,out=315,in=225,loop] node[below] {$x$} (1);
\path (2) edge[out=315,in=225,loop] node[below] {$3$} (2);
\path (3) edge[in=315,out=225,loop] node[below] {$y$} (3);
\path (2) edge[ultra thick,->] (a);
\path (a) edge[->] (b1);
\path (a) edge[ultra thick,->] (b2);
\path (b1) edge[->] (c);
\path (b2) edge[ultra thick,->] (c);
\path (c) edge[ultra thick,->] (d);
\path (d) edge[ultra thick,->] (3);
\path (3) edge[ultra thick,bend left,->] node[below] {$x$} (2);
\path (a) edge[out=135,in=45,loop] (a);
\path (b1) edge[ultra thick,out=135,in=45,loop] (b1);
\path (b2) edge[out=315,in=225,loop] (b2);
\path (c) edge[out=135,in=45,loop] (c);
\path (d) edge[out=135,in=45,loop] (d);
\end{tikzpicture}
};
\node[draw, rectangle] (2cyclesa) at (1,2.4) {
\begin{tikzpicture}[scale=1]
\node[draw, circle,inner sep=0pt,minimum size=0.2cm] (1) at (1,0) {};
\node[draw, circle,inner sep=0pt,minimum size=0.2cm] (2) at (2,0) {};
\node[draw, circle,inner sep=0pt,minimum size=0.2cm] (3) at (5,0) {};
\node[draw, circle,inner sep=0pt,minimum size=0.2cm] (a) at (2,1) {};
\node[draw, circle,inner sep=0pt,minimum size=0.2cm] (b1) at (3,1.5) {};
\node[draw, circle,inner sep=0pt,minimum size=0.2cm] (b2) at (3,0.5) {};
\node[draw, circle,inner sep=0pt,minimum size=0.2cm] (c) at (4,1) {};
\node[draw, circle,inner sep=0pt,minimum size=0.2cm] (d) at (5,1) {};
\path (1) edge[ultra thick,out=315,in=225,loop] node[below] {$x$} (1);
\path (2) edge[out=315,in=225,loop] node[below] {$3$} (2);
\path (3) edge[in=315,out=225,loop] node[below] {$y$} (3);
\path (2) edge[ultra thick,->] (a);
\path (a) edge[ultra thick,->] (b1);
\path (a) edge[->] (b2);
\path (b1) edge[ultra thick,->] (c);
\path (b2) edge[->] (c);
\path (c) edge[ultra thick,->] (d);
\path (d) edge[ultra thick,->] (3);
\path (3) edge[ultra thick,bend left,->] node[below] {$x$} (2);
\path (a) edge[out=135,in=45,loop] (a);
\path (b1) edge[out=135,in=45,loop] (b1);
\path (b2) edge[ultra thick,out=315,in=225,loop] (b2);
\path (c) edge[out=135,in=45,loop] (c);
\path (d) edge[out=135,in=45,loop] (d);
\end{tikzpicture}
};
\draw (2cycles) -> (2cyclesa);
\draw (2cycles) -> (2cyclesb);
\draw (3cycles) -> (3cyclesa);
\node (matrix) at (0,1.5) {
\begin{minipage}{5cm}
$C=\begin{pmatrix}
3 & 0 & -2 \\
0 & x & 0 \\
x & 0 & y
\end{pmatrix}$\\
\mbox{~}\\
$\det(C)=3xy+2x^2$
\end{minipage}
};
\end{tikzpicture}
\caption{Given a matrix $C$ we construct a digraph~$H$ with directed adjacency matrix $C$ (left hand side) and the digraph $G$ (right hand side) by replacing the edge with label $-2$ in $H$ by a binary algebraic branching program.
We omit the labels for edges that have label~$1$.
The right hand side depicts the cycle covers $K$ of~$G$ and the left hand side shows the corresponding cycle covers $K^H$ of~$H$.}
\label{fig:cyclecovers}
\end{figure}
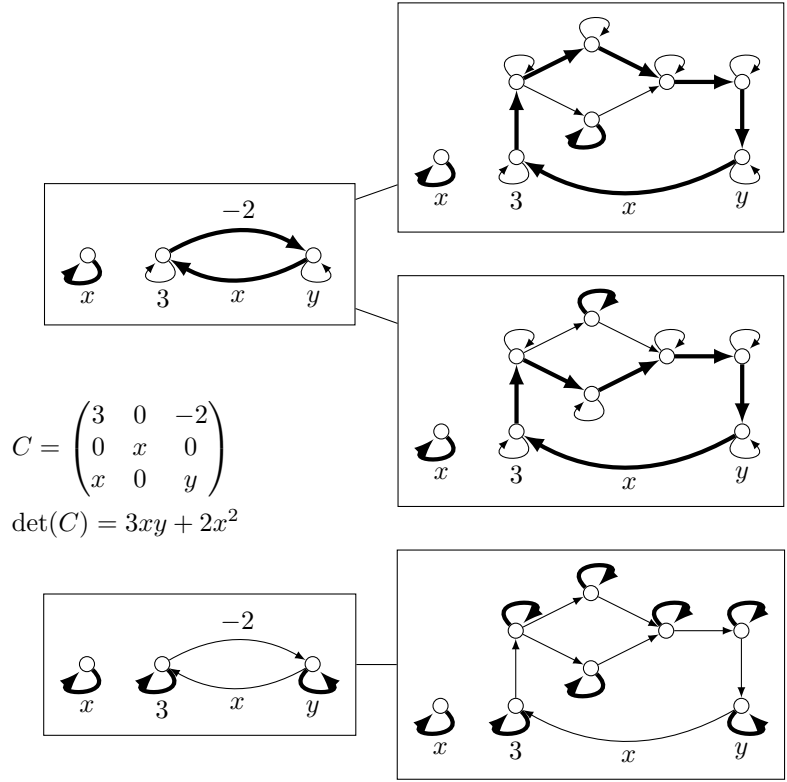
We will interpret $C$ as the directed adjacency matrix of a digraph. Any edge that has an integer label which is neither $1$ nor $0$ will be replaced by a subgraph of size $\gO(\log|c_{\max}|)$ arising from the construction of the previous Lemma \ref{prop:constants}. The directed adjacency matrix of the resulting graph will be the desired matrix~$A$. Formally, we proceed by induction.

Denote by $q$ the number of integer entries in the matrix $C$ that are neither equal to $0$ nor~$1$. By induction on~$q$, we will prove the slightly stronger statement that there is a binary variable matrix~$A$ of size $n + q\cdot\gO(\log|c_{\max}| )$ with $\det(A)=\det(C)$. Since $q \le n^2$, this implies the statement. Note that the case $q=0$ is trivial, so we assume $q\geq 1$ and perform the induction step.

Let $H$ be the digraph whose directed adjacency matrix is~$C$. Recall that this means the following: $H$ is a digraph on the vertices $1\dts,n$ and there is an edge $(i,j)$ with label $C_{ij}$ if $C_{ij}\ne0$ and otherwise no such edge exists. Let $e=(i,j)$ be the edge corresponding to an integer entry $c=C_{ij}$ which is neither $0$ nor~$1$. Let $\Gamma=(\Gamma,s,t)$ be a binary algebraic branching program with path value $c$ and $\gO(\log|c|)$ many vertices, which exists by Proposition \ref{prop:constants}.

We will now replace the edge $(i,j)$ by $\Gamma$ (see Figure~\ref{fig:cyclecovers}): Let $G$ be the digraph that arises from $H\cup \Gamma$ by removing the edge $(i,j)$, adding edges $(i,s)$ and $(t,j)$ with label $1$ and adding loops with label $1$ to all vertices of~$\Gamma$. The directed adjacency matrix of $G$ has size $n+\gO(\log|c|)\le n+\gO(\log|c_{\max}|)$ and contains $q-1$ integer entries which are neither $0$ nor~$1$. By applying the induction hypothesis to the directed adjacency matrix of~$G$, we obtain a matrix $A$ of size
\[ n + \gO(\log|c_{\max}|) + (q-1)\cdot\gO(\log|c_{\max}|)
= n + q\cdot\gO(\log|c_{\max}|) \]
whose determinant equals the value of~$G$. We are left to show that the value of $G$ is equal to $\det(C)$, i.e., the value of~$H$.

For this purpose, we will analyze the relation between cycle covers of $G$ and~$H$, which is straightforward (see Figure~\ref{fig:cyclecovers}): Consider a cycle cover $K$ of~$G$. Any vertex of $\Gamma$ which is not covered by its loop must be part of a cycle whose intersection with $\Gamma$ is a path from $s$ to~$t$. To $K$ we can therefore associate a cycle cover $K^H$ of $H$ as follows: If every vertex of $\Gamma$ is covered by its loop in~$K$, let $K^H$ be $K$ without these loops. Otherwise, there is unique cycle $\kappa_K$ in $K$ that restricts to an $s$-$t$-path $\pi_K$ in~$\Gamma$. Let $\kappa_K^H$ be the intersection $\kappa_K\cap H$ together with the edge $(i,j)$ and note that $\kappa_K^H$ is a cycle in~$H$.
We obtain $K^H$ from $K$ by replacing $\kappa_K$ with $\kappa^H_K$ and removing all remaining loops from inside~$\Gamma$.

All cycle covers $L$ of $H$ are of the form $L=K^H$ for some cycle cover $K$ of~$G$. If $L$ is a cycle cover of $H$ containing the edge $(i,j)$ then the cycle covers $K$ of $G$ with $L=K^H$ are in bijection with the $s$-$t$-paths in~$\Gamma$. We now fix such a cycle cover~$L$. By definition of the value of a digraph, it suffices to show that
\[ \sum_{\begin{subarray}{c}K\text{ cycle cover of }G\\\text{such that }L=K^H\end{subarray}} \weight(K) = \weight(L). \]
Note that $K$ and $L=K^H$ differ only in loops and in the cycles $\kappa_K$ and $\kappa_K^H$, respectively. Since loops contribute a factor of $1$ to the weight of a cycle cover, we are left to prove that
\[ \sum_{\begin{subarray}{c}K\text{ cycle cover of }G\\\text{such that }L=K^H\end{subarray}} \weight(\kappa_K) = \weight(\kappa_K^H). \]
Let $e_1\dts,e_r$ be the labels of the edges of $\kappa_K\cap H$. These are the edges shared by $\kappa_K$ and $\kappa_K^H$. Thus,
\begin{align*}
\weight&(\kappa_K^H)
= (-1)^{r} \cdot c \cdot e_1\cdots e_r
\\ &= \left( \sum_{
\begin{subarray}{c}
\pi\text{ is $s$-$t$-path} \\
\text{inside }P
\end{subarray}}
\weight(\pi) \right) \cdot (-1)^{r} \cdot e_1\cdots e_r
= \left( \sum_{
\begin{subarray}{c}
K\text{ cycle cover of }G \\
\text{such that }L=K^H\end{subarray}}
\weight(\pi_K) \right) \cdot (-1)^{r} \cdot e_1\cdots e_r
\\ &= \left( \sum_{
\begin{subarray}{c}
K\text{ cycle cover of }G \\
\text{such that }L=K^H\end{subarray}}
\weight(\pi_K)  \cdot (-1)^{r} \cdot e_1\cdots e_r\right)
= \sum_{\begin{subarray}{c}K\text{ cycle cover of }G\\\text{such that }L=K^H\end{subarray}} \weight(\kappa_K)
\end{align*}
is precisely the desired equality.
\end{proof}

\begin{corollary}
\label{cor:valuniversality}
For every polynomial $f$ with integer coefficients there exists a binary variable matrix whose determinant is~$f$.
\end{corollary}
\begin{proof}
Combine Theorem~\ref{thm:universalityofdet} and Proposition~\ref{pro:integercoeffs}.
\end{proof}

\section{Lower Bounds}
\label{sec:lower}

This section is dedicated to the proof of Theorem \ref{PermanentLowerBound}.
Let $\HB\df\set{0,1}$.
A sequential numbering makes the proof much easier to read, so we think of the variables as arranged in a $3\times 3$ matrix
\[ x = \begin{pmatrix}
x_1 & x_2 & x_3 \\
x_4 & x_5 & x_6 \\
x_7 & x_8 & x_9
\end{pmatrix}. \]
In this section, we will understand $\per_3=\per(x)$ as a polynomial in the variables $x_1\dts,x_9$ instead of the variables $x_{ij}$ with $1\le i,j\le 3$.

\subsection*{Proof Outline}

Let $n\in\HN$ and $A$ an $n\times n$ binary variable matrix.
The binary matrix $B(A)\in\HB^{n\times n}$ is defined as the matrix arising from $A$ by setting all variables to~$1$. We call $B(A)$ the \dd{support matrix} of~$A$. If we set all variables to $1$ in $\per_3$, we obtain the value~$6$, so if $\per_3=\det(A)$, then substituting~$1$ for all variables on both sides of the equation, we obtain the condition
\begin{equation}
\label{theSixCondition}
6=\det(B(A)).
\end{equation}
In \cite{ZAMM,HADAMARDSEQUENCE}, the maximal values of determinants of binary matrices are computed for small values of~$n$. Since
\begin{equation}
\label{Hadamard}
\forall B\in\HB^{5\times5}\colon\:\det(B)\le5,
\end{equation}
we immediately obtain the lower bound $\bdc(\per_3)\ge 6$.

Unfortunately, there are several matrices $B\in\HB^{6\times6}$ that satisfy $\det(B)=6$. We proceed in two steps to verify that nevertheless, none of these matrices $B$ is the support matrix $B(A)$ of a candidate matrix $A$ with $\per_3=\det(A)$. A rough outline is the following:

\begin{enumerate}[label=(\alph*),leftmargin=*,ref=part (\alph*)]
\item Enumerate all matrices $B\in\HB^{6\times6}$ with $\det(B)=6$ up to symmetries.
\item \label{secondpart} For all those matrices $B$ prove that $B$ is not the support matrix~$B(A)$ of a binary variable matrix~$A$ with $\det(A)=\per_3$. We describe this process in the next subsection.
\end{enumerate}

\subsection{Stepwise Reconstruction}
\label{stepwise}
Let us make \ref{secondpart} precise. In the hope of failing, we attempt to reconstruct a binary variable matrix~$A$ that has support~$B$ and which also satisfies $\det(A)=\per_3$. During the reconstruction process, we successively replace $1$'s in $B$ by the next variable. The process is as follows:

Given a binary matrix $B\in\HB^{6\times 6}$, let 
\[ S\df\cset{ (i,j) }{ B_{ij}=1 } \]
be the set of possible variable positions. For any set of positions $I\subseteq S$, we consider the matrix $B_{I}$ that arises from $B$ by placing a variable $y$ in every position in~$I$. If $B$ is the support of a binary variable matrix $A$ with $\det(A)=\per_3$ and $I$ contains exactly the positions where $y := x_1$ occurs in~$A$, then $\det(B_{I})$ must be equal to
\begin{equation}
\label{Per3withY}
\per_3\begin{pmatrix}
y & 1 & 1 \\
1 & 1 & 1 \\
1 & 1 & 1
\end{pmatrix} = 2y + 4.
\end{equation}
We define the set 
\[ \mathcal S\df \cset{ I\subseteq S }{ \det(B_I) = 2y + 4 }. \]

\begin{claim} Let $A$ be a binary variable matrix with support $B$ and $\det(A)=\per_3$. Let $k\in\set{1\dts,9}$ and define $I_k\df\cset{ (i,j) }{ A_{ij} = x_k }$ as the set of positions where the variable $x_k$ occurs in~$A$. Then, we have $I_k\in\mathcal{S}$.
\end{claim}
\begin{proof} By the symmetry of the permanent, we may assume that $k=1$. In the matrix~$A$, setting every variable except $y\df x_1$ to $1$ yields the matrix $B_I$ and therefore, $\det(B_I)=2y+4$ as in \eqref{Per3withY}, because $\det(A)=\per_3$. This means $I_k\in\mathcal{S}$ by definition.
\end{proof}

Therefore, if $B$ is the support matrix $B(A)$ of a binary variable matrix $A$ with $\det(A)=\per_3$, we can find $9$ pairwise disjoint sets in $\mathcal S$, one for each variable $x_k$, that specify precisely where to place these variables in~$A$.

By a recursive search and backtracking, we now look for sets $I_1\dts,I_k\in\mathcal S$ such that
\begin{enumerate}[label=(\roman*)]
\item $I_1, \ldots, I_k$ are pairwise disjoint.
\item Placing $x_i$ into $B$ at every position from $I_i$ for $1\le i\le k$ yields a matrix $A_k$ such that $\det(A_k)$ is equal to $\per_3(x_1\dts,x_k,1\dts,1)$.
\end{enumerate}
The search is recursive in the following sense: First, the possible choices at depth $k=1$ are given by $\mathcal S$. Enumerating the possible choices for depth $k+1$ works as follows: For each choice $I_1\dts,I_k\in\mathcal S$ with the above two properties, we enumerate all $I_{k+1}\in\mathcal S$ that have empty intersection with $I_1\pts\cup I_{k}$ and check whether condition~(b) is satisfied.

If the recursive search never reaches $k=9$ or fails there, then $B$ is not the support of a binary variable matrix $A$ with $\det(A)=\per_3$. If we reach level $9$ however and do not fail there, we have found such an~$A$.

In practice, the process is sped up significantly by working over a large finite field $\mathbb F_p$ and choosing random elements $x_1\dts,x_9\in\mathbb F_p\setminus\set{0,1}$.

\subsection{Exploiting Symmetries in Enumeration}
Let us call two matrices \emph{equivalent} if they arise from each other by transposition and/or permutation of rows and/or columns.
A key observation is that equivalent matrices have the same determinant up to sign.
Therefore we do not have to list all binary matrices $B \in \HB^{6 \times 6}$ with $\det(B)=6$,
but it suffices to list one representative matrix $B$ with $\det(B)=\pm 6$ for each equivalence class.
It happens to be the case that the equivalence classes of $6 \times 6$ binary matrices are in bijection to graph isomorphy classes of undirected bipartite graphs
$G=(V\cup W,E)$ with $\abs V=\abs W=6$, $V\cap W=\emptyset$ as follows:
For $V=\set{v_1\dts,v_6}$ and $W=\set{w_1\dts,w_6}$, the bipartite adjacency matrix $B(G)\in\HB^{6\times 6}$ of $G$ is defined via $B(G)_{i,j}=1$ if and only if $\set{v_i,w_j}\in E$.
Row and column permutations in $B(G)$ are reflected by renaming vertices in~$G$.
Transposition of $B(G)$ amounts to switching $V$ and $W$ in~$G$.

The computer software nauty \cite{NAUTY} can enumerate all $251\,610$ of these bipartite graphs, which is already a significant improvement over the $2^{36}=68\,719\,476\,736$ elements of $\HB^{6\times6}$.
To further limit the number of bipartite graphs that have to be considered, we make the following observations: 

\begin{itemize}[leftmargin=*]
\item We need not consider binary matrices $B$ containing a row $i$ with only a single entry $B_{ij}$ equal to~$1$. Indeed, Laplace expansion over the $i$-th row yields that $\det(B)$ is equal to the determinant of a $5\times 5$ binary matrix, which can at most be~$5$, see \eqref{Hadamard}.
Translating to bipartite graphs, we only need to consider those bipartite graphs where all vertices have at least two neighbours.
\item If two distinct vertices in $G$ have the same neighbourhood, then the bipartite adjacency matrix $B(G)$ has two identical rows (or columns) which would imply $\det(B(G))=0$. Hence, we only need to enumerate bipartite graphs where all vertices have distinct neighbourhoods. Unfortunately nauty can impose this restriction only on rows and not on columns.
\end{itemize}

\noindent With these restrictions, the nauty command
\begin{center}
\texttt{genbg -d2:2 -z 6 6}
\end{center}
generates $44\,384$ bipartite graphs, only $263$ of which have a bipartite adjacency matrix with determinant equal to $\pm 6$. We then preprocess this list by swapping the first two rows of any matrix with negative determinant.

Finally, the stepwise reconstruction (section \ref{stepwise}) fails for all of these $263$ matrices, proving that $\bdc(\per_3)\ge 7$.
The algorithm takes 28 seconds on an Intel Core\texttrademark{} i7-4500U CPU (2.4 GHz) to finish.

\medskip Unfortunately, $\bdc(\per_4)$ can currently not be determined in this fashion because the enumeration of all apropriate bipartite graphs, already on $9+9$ vertices, is infeasible.

\section{Uniqueness of the Grenet construction in the 7 by 7 case}
\label{sec:uniqueness}
The methods from Section~\ref{sec:lower} can be used to determine all $7\times7$ binary variable matrices $A$ with the property that $\det(A)=\per_3$. By means of a cluster computation over the course of one week, we determined all $463$ binary variable matrices with this property and made some noteworthy discoveries.

The Grenet construction (see Section~\ref{grenetconstruction}) yields the matrix 
\begin{equation} \label{grenet7x7} \begin{pmatrix}
 x_{11} & x_{12} & x_{13} & 0 & 0 & 0 & 0 \\
 1 & 0 & 0 & x_{32} & x_{33} & 0 & 0 \\ 
 0 & 1 & 0 & x_{31} & 0 & x_{33} & 0 \\ 
 0 & 0 & 1 & 0 & x_{31} & x_{32} & 0 \\
 0 & 0 & 0 & 1 & 0 & 0 & x_{23} \\ 
 0 & 0 & 0 & 0 & 1 & 0 & x_{22} \\
 0 & 0 & 0 & 0 & 0 & 1 & x_{21}
\end{pmatrix}. \end{equation}
It is the unique \enquote{sparse} $7\times7$ binary variable matrix from among the $463$, in the sense that every other matrix from the list has more than three nonzero entries in some row or column. 

Motivated by the above observation, we verified by hand (with computer support) that in fact, all of the $463$ matrices can be reduced to \eqref{grenet7x7} by means of elementary row and column operations. This can be summarized as follows:

\begin{proposition}
Every $7\times 7$ binary variable matrix $A$ with $\det(A)=\per_3$ is equivalent to the Grenet construction \eqref{grenet7x7} under the following two group actions:
\begin{enumerate}[leftmargin=*]
\item The action of  $\cset{ (g,h)}{ \det(g)=\det(h) }\subseteq \operatorname{GL}_7(\HZ)\times\operatorname{GL}_7(\HZ)$ on $7\times7$ matrices via left and right multiplication, together with transposition of $7\times7$ matrices. 
\item The action of $\kS_3\times\kS_3$ on the variables $x_{ij}$ with $1\le i,j\le 3$, and the corresponding transposition (i.e. the map $x_{ij}\mapsto x_{ji}$.) 
\end{enumerate} 
Note that (1) leaves the determinant of any $7\times7$ binary variable matrix invariant and (2) leaves the permanent polynomial invariant.
\qed 
\end{proposition}

\begin{example}
One of the matrices that occur in our enumeration is the matrix
\[ 
A \df \begin{pmatrix}
x_{31} &x_{32} &x_{31} &0 &x_{32} &1 &x_{23} \\
 1 &x_{33} &0 &x_{31} &x_{33} &x_{31} &x_{22} \\
x_{33} &0 &x_{33} &x_{32} &1 &x_{32} &x_{21} \\
 1 &0 &1 &0 &0 &0 &x_{22} \\
 0 &x_{11} &x_{12} &x_{13} &0 &0 &0 \\
 0 &1 &0 &0 &1 &0 &x_{21} \\
 0 &0 &0 &1 &0 &1 &x_{23}
\end{pmatrix}.
\]
One can check that indeed $\det(A)=\per_3$. In this case, the matrices 
\begin{align*}
g &\df \begin{pmatrix}
0 & 0 & 0 &  0 & -1 &  0 &  0 \\
0 & 0 & 1 &  0 &  0 & -1 &  0 \\
0 & 1 & 0 & -1 &  0 &  0 &  0 \\
1 & 0 & 0 &  0 &  0 &  0 & -1 \\
0 & 0 & 0 &  0 &  0 &  0 &  1 \\
0 & 0 & 0 &  1 &  0 &  0 &  0 \\
0 & 0 & 0 &  0 &  0 &  1 &  0
\end{pmatrix}, & 
h &\df \begin{pmatrix} 
 0 &  1 &  0 & 0 & 1 & 0 & 0 \\
-1 &  0 &  0 & 0 & 0 & 0 & 0 \\
 0 & -1 &  0 & 0 & 0 & 0 & 0 \\
 0 &  0 & -1 & 0 & 0 & 0 & 0 \\
 1 &  0 &  0 & 0 & 0 & 1 & 0 \\
 0 &  0 &  1 & 1 & 0 & 0 & 0 \\
 0 &  0 &  0 & 0 & 0 & 0 & 1
\end{pmatrix}
\end{align*} 
are both invertible over $\HZ$ and $gAh$ is precisely \eqref{grenet7x7}.
\end{example}

\section{Algebraic Complexity Classes}
\label{sec:vpws}
In this section we relate binary determinantal complexity to classical complexity measures.
An \emph{algebraic circuit} $\mathcal C$ over the rational numbers is a directed acyclic digraph whose vertices have indegree 0 or 2
with a single vertex having outdegree 0.
Those with indegree 0 are labeled with an integer or a variable, and are called \emph{input gates}.
Those with indegree 2 are labeled with either $+$ or $\times$ and are called \emph{addition gates} and \emph{multiplication gates}, respectively.
At each addition or multiplication gate the circuit $\mathcal C$ defines a polynomial with rational coefficients via adding/multiplying
the polynomials of its two parents.
For the polynomial $f$ which is defined at the unique vertex with outdegree~0 we say that the circuit \emph{$\mathcal C$ computes $f$}.
If all input gates are labeled with either~$1$, $-1$, or a variable, the circuit is called \emph{constant-free}.
Note that every constant-free circuit computes a polynomial that has integer coefficients.
An algebraic circuit is called \emph{skew} if for every multiplication gate at least one of its two parents is an input gate.
An algebraic circuit is called \emph{weakly skew} if for every multiplication gate $\alpha$ there is at least one of its two parents $\beta$ for which the circuit graph splits into disjoint connected components if we remove the edge between $\alpha$ and~$\beta$.
The \emph{skew complexity} of $f$ is defined as the minimal number of vertices required for a skew algebraic circuit to compute~$f$.
Analogously, the \emph{weakly skew complexity} of $f$ is defined as the minimal number of vertices required for a weakly skew algebraic circuit to compute~$f$.
Moreover, the \emph{constant-free skew complexity} of $f$ is defined as the minimal number of vertices required for a constant-free skew algebraic circuit to compute~$f$
and the \emph{constant-free weakly skew complexity} of $f$ is defined as the minimal number of vertices required for a constant-free weakly skew algebraic circuit to compute~$f$.
The complexity class $\VPs$ is defined as the set of sequences of polynomials with polynomially bounded skew complexity.
Analogously, the complexity class $\VPws$ is defined as the set of sequences of polynomials with polynomially bounded weakly skew complexity,
the complexity class $\VPs^0$ is defined as the set of sequences of polynomials with polynomially bounded constant-free skew complexity,
and the complexity class $\VPws^0$ is defined as the set of sequences of polynomials with polynomially bounded constant-free weakly skew complexity, see also \cite{Mal:03}.
A fundamental result in \cite{toda:92} (see also \cite{mapo:08}) is that $\VPws=\VPs$.
Analyzing the constants which appear in the proof of $\VPws=\VPs$ in \cite{toda:92}, we see that the proof immediately yields $\VPws^0=\VPs^0$.
For the sake of comparison with $\VPs^0$, let us make the following definition.
\begin{definition}
The complexity class $\deto$ consists of all sequences of polynomials that have polynomially bounded binary determinantal complexity $\bdc$.
\end{definition}
The main purpose of this section is to show the following statement.
\begin{proposition}
$\VPs^0 = \deto$.
\end{proposition}
\begin{proof}
The proof of \cite[Lemma 3.4]{toda:92} immediately shows that $\deto \subseteq \VPs^0$.
To show that $\VPs^0 \subseteq \deto$ we want to adapt the proof of \cite[Lemma 3.5 or Theorem 4.3]{toda:92},
but a subtlety arises:
The proof shows that from a weakly skew or skew circuit $\mathcal C$ we can construct a matrix $A'$
of size polynomially bounded in the number of vertices in $\mathcal C$ such that $\det(A')$ is the polynomial compute by $\mathcal C$
with the drawback that $A'$ is not a binary variable matrix, but $A'$ has as entries variables and constants~$0$,~$1$, and~$-1$.
Fortunately Proposition~\ref{pro:integercoeffs} establishes $\deto = \VPs^0 = \VPws^0$. 
\end{proof}

\begin{remark}
In the past, other models of computation with bounded coefficients have already given way to stronger lower bounds than their corresponding unrestricted models: \cite{Mor:73} on the fast fourier transform, \cite{Raz:03} on matrix multiplication, and \cite{BL:04} on arithmetic operations on polynomials.
\end{remark}

From Valiant's completeness result \cite{Val:79b} we deduce that $\VP \neq \VNP$ implies $\per_m \notin \VPws^0$.
A main goal
is to prove $\per_m \notin \VPws^0$ unconditionally.
This could be a simpler question than $\VP \neq \VNP$ or even $\VP^0 \neq \VNP^0$,
because with what is known today, from $\per_m \in \VPws^0$ we cannot conclude $\VP^0 = \VNP^0$, see \cite[Thm.~4.3]{Koi:04}.
If we replace the permanent polynomial by the Hamiltonian Cycle polynomial 
\[ \operatorname{HC}_m \df \sum_{
\begin{subarray}{c}\pi\in\kS_m\\\pi\text{ is $m$-cycle}\end{subarray}} \prod_{i=1}^m x_{i,\pi(i)},
\]
then the question $\text{HC}_m \notin \VPws$ is indeed equivalent to separating $\VPws^0$ from $\VNP^0$, see \cite[Thm.~2.5]{Koi:04}, mutatis mutandis.
We ran our analysis for $\text{HC}_m$, $m\le4$ and proved $\bdc(\text{HC}_1)=1$, $\bdc(\text{HC}_2)=2$, $\bdc(\text{HC}_3)=3$, $\bdc(\text{HC}_4)\geq 7$.
This means that $7 \leq \bdc(\text{HC}_4)\leq 13$,
where the upper bound follows from considerations analogous to Grenet's construction, see Section~\ref{subsec:hc}.

\section{Graph Constructions for Polynomials}

In this section, we review the proof of Theorem \ref{thm:grenet} from \cite{Gre:11}. Furthermore, we use the same methods to prove the following result about the Hamiltonian Cycle polynomial:

\begin{theorem} \label{hc:bdc} For all natural numbers $m\in\HN$, we have $\bdc(\operatorname{HC}_{m+1})\le  m\cdot 2^{m-1} + 1$.
\end{theorem}

In this section, we denote by $[m]\df\set{1\dts,m}$ the set of numbers between $1$ and~$m$. 

\subsection{Grenet's Construction for the Permanent}
\label{grenetconstruction} In this subsection, we prove Theorem \ref{thm:grenet}. The construction of Grenet is a digraph $\Gamma$ whose vertices $V\df\cset{ v_I }{ I\subseteq[m] }$ are indexed by the subsets of $[m]$. Hence, $|V|=2^m$. We partition $V=V_0\cup\cdots\cup V_m$ such that $V_i$ contains the vertices belonging to subsets of size~$i$. We set $s\df v_\emptyset$ and $t\df v_{[m]}$, so $V_0=\set s$ and $V_m=\set t$. Edges will go exclusively from $V_{i-1}$ to $V_{i}$ for $i\in[m]$. In fact, we insert an edge from $v_{I}$ to $v_{J}$ if and only if there is some $j\in[m]$ with $J=I\cup\set j$. This edge is then labeled with the variable $x_{ij}$, where $i=|J|$. For example, there are $m$ edges going from $V_0$ to $V_1$, one for each variable $x_{1j}$ with $1\le j\le m$. It is clear that for each permutation $\pi\in\kS_m$, there is precisely one $s$-$t$-path in $\Gamma$ whose path weight is $(-1)^{m-1}\cdot x_{1,\pi(1)}\cdots x_{m,\pi(m)}$. Consequently, the path value of the algebraic branching program $\Gamma=(\Gamma,s,t)$ is equal to $(-1)^{m-1}\cdot\per_m$. Theorem \ref{thm:grenet} then follows from the following lemma:

\begin{lemma} \label{BranchingProgramToBDC} Let $\Gamma=(\Gamma,s,t)$ be a binary algebraic branching program on $n\ge 3$ vertices with path value $\pm f$. Then, there is a binary variable matrix of size $n-1$ whose determinant is equal to~$f$.
\end{lemma}
\begin{proof} We first construct a graph $G$ from $\Gamma$ by identifying the two vertices $s$ and $t$ and adding loops with label $1$ to every other vertex. The $s$-$t$-paths in $\Gamma$ are then in one-to-one correspondence with the cycle covers of $G$: Indeed, any cycle cover in $G$ must cover the vertex $s=t$ and this cycle corresponds to an $s$-$t$-path in~$\Gamma$. Every other vertex can only be covered by its loop because $\Gamma$ is acyclic. The graph $G$ now has the value $\pm f$ by definition and its directed adjacency matrix $A$ has size $n-1$. Since $n-1\ge 2$, we can exchange the first two rows of $A$ to change the sign of its determinant.
\end{proof}

\subsection{Hamiltonian Cycle Polynomial}\label{subsec:hc}
In this subsection, we prove Theorem \ref{hc:bdc} using Lemma \ref{BranchingProgramToBDC}.
In order to construct a binary algebraic branching program $\Gamma=(\Gamma,s,t)$ with path value $\operatorname{HC}_{m+1}$, we proceed similar to subsection \ref{grenetconstruction}. We will refer to cyclic permutations in $\kS_{m+1}$ of order $m+1$ simply as \emph{cycles} because no cyclic permutations of lower order will be considered. Observe that the cycles in $\kS_{m+1}$ are in bijection with the permutations in $\kS_m$. This can be seen by associating to $\pi\in\kS_m$ the cycle
$\sigma=(\pi(1)\dts,\pi(m),m+1)\in\kS_{m+1}$.
In other words, $\sigma$ maps $m+1$ to $\pi(1)$, it maps $\pi(1)$ to $\pi(2)$ and so on.

In addition to two vertices $s$ and~$t$, our binary algebraic branching program will have a vertex $v_{(I,i)}$ for every nonempty subset $I\subseteq[m]$ and $i\in I$. 
By our above Lemma \ref{BranchingProgramToBDC}, the resulting binary variable matrix will have a size of
\[
1 + \sum_{i=1}^{m} \binom{m}i \cdot i = m\cdot 2^{m-1} + 1.
\]
For $m=3$, this is equal to $3\cdot 2^2 + 1 = 13$. 

We will construct the edges in $\Gamma$ in such a way that every cycle $\sigma=(a_1\dts,a_m,m+1)$ corresponds to an $s$-$t$-path which has $v_{(I,i)}$ as its $k$-th vertex if and only if $I=\set{a_1\dts,a_k}$ and $i=a_k$. We insert the following edges:
\begin{itemize}[leftmargin=*]
\item from $s$ to $v_{(\set i,i)}$ for each $i\in[m]$ with label $x_{m+1,i}$
\item from $v_{(I,i)}$ to $v_{(I\cup\set j, j)}$ for each $i\in I\subseteq[m]$ and $j\in[m]\setminus I$ with label $x_{i,j}$
\item from $v_{([m],i)}$ to $t$ for each $i\in[m]$ with label $x_{i,m+1}$. 
\end{itemize}

We can again partition the set of vertices as $V=V_0\cup V_1\cup\cdots\cup V_{m+1}$ where $V_0=\set s$, $V_{m+1}=\set t$ and for $k\in[m]$, the set $V_k$ consists of all vertices $v_{(I,i)}$ with $\abs I=k$. Then, edges go only from $V_k$ to~$V_{k+1}$, in particular $\Gamma$ is acyclic. Furthermore, all $s$-$t$-paths in $\Gamma$ have the same lengths and correspond uniquely to cycles in $\kS_{m+1}$ . This concludes the proof of Theorem \ref{hc:bdc}.

\medskip We know no better construction for arbitrary~$m$, but for small $m$ we have 
\begin{align*}
\operatorname{HC}_2 &=
\det\begin{pmatrix}
x_{12} & 0 \\
0 & x_{21}
\end{pmatrix} &
\operatorname{HC}_3 &=
\det\begin{pmatrix}
0 & x_{12} & x_{13} \\
x_{21} & 0 & x_{23} \\
x_{31} & x_{32} & 0
\end{pmatrix}.
\end{align*}

\bibliographystyle{alpha}

\end{document}